\documentclass[11pt, a4paper]{article}
\usepackage[T2A]{fontenc}
\usepackage[utf8]{inputenc}
\usepackage[english]{babel}
\usepackage{amsmath, amssymb, amsthm}
\usepackage{geometry}
\usepackage{caption}
\usepackage{verbatim}
\usepackage{float}
\usepackage{amsmath, amsthm, array, booktabs}
\usepackage{enumitem}
\usepackage{geometry}
\usepackage{url}
\usepackage{hyperref}
\usepackage{graphicx}
\usepackage{cite}
\newtheorem{theorem}{Theorem}[section]
\newtheorem{lemma}[theorem]{Lemma}
\newtheorem{proposition}[theorem]{Proposition}

\DeclareMathOperator{\NN}{\mathbb{N}}
\DeclareMathOperator{\FF}{\mathbb{F}}

\title{Linear Binary Codes Correcting One or More Errors}
\author{
  Timofei Izhitskii\\
  \textit{faculty of computer science}\\
  HSE University\\
  Moscow, Russia\\
  \texttt{tmizhitskii@edu.hse.ru}
}
\date{}

\begin{document}

\maketitle

\begin{abstract}
This paper examines linear binary codes capable of correcting one or more errors. For the single‐error‐correcting case, it is shown that the Hamming bound is achieved by a constructive method, and an exact expression for the minimal codeword length is derived. For the general case, a simple lower bound for the parameters of linear codes is derived from an analysis of the coset structure.

\end{abstract}

\begin{center}
    \vspace{-0.5em}
    \normalsize
\end{center}

\section{Introduction}
Let 
\(
F : \FF_2^k \longrightarrow \FF_2^n
\)
be a binary code with minimum distance $d$. 
The code can correct up to $\lfloor (d-1)/2 \rfloor$ errors. 
The mapping $F$ represents a code with parameters $(n,k,d)_2$ in general; 
when the code is linear, we write $[n,k,d]_2$. In this work the emphasis is on the linear case~\cite{MacWilliamsSloane,ref1}.

Denote by $N(k,e)$ the minimal value of $n$ for which there exists (possibly nonlinear) code correcting $e$ errors, and by $\mathrm{NL}(k,e)$ the minimal length $n$ among all linear codes. 
Clearly, $\mathrm{NL}(k,e) \ge N(k,e)$.

There exists an example where $N(k,e) < \mathrm{NL}(k,e)$. For $(k,e)=(8,2)$ the Nordstrom--Robinson code is a nonlinear $(16,8,6)$ binary code~\cite{NR67}.
It yields a nonlinear $(15,8,5)$ code (see also~\cite{MacWilliamsSloane}), hence $N(8,2)\le 15$. 
Moreover, by~\cite{Grassl,MacWilliamsSloane} one has $\mathrm{NL}(8,2)>15$.
Therefore,
\[
N(8,2)\le 15 < \mathrm{NL}(8,2).
\]

The first part of the paper is devoted to the case $e=1$.  
It is shown that the Hamming bound is attained for all $k$, that is,
\[
\mathrm{Ham}(k,1)=N(k,1),
\]
and an explicit closed formula is obtained:
\[
\mathrm{NL}(k,1)=N(k,1)
= k+1+\left\lfloor \log_2\!\bigl(k+1+\lfloor \log_2 k\rfloor\bigr)\right\rfloor.
\]

The second part focuses on linear codes for arbitrary $e \ge 1$.
An analysis of the coset structure of an $[n,k,2e+1]_2$ code leads to Theorem~\ref{th:tilted} and yields a computable lower bound
\[
\mathrm{NL}(k,e)\ge C(k,e).
\]
For most values of $k$ this coincides with the Hamming bound $\mathrm{Ham}(k,e)$, whereas for certain pairs $(k,e)$ it is strictly stronger; see Tables~\ref{tab:C2} and~\ref{tab:C3}.

\section{Preliminaries}

For $u,v \in \FF_2^n$, their \emph{Hamming distance} is defined by
\[
\rho(u,v) := |\{\,i \mid u_i \neq v_i\,\}|,
\]
that is, the number of coordinates in which $u$ and $v$ differ.  
The \emph{weight} of a vector $u$ is 
\[
w(u) := \rho(u,0),
\]
that is, the number of nonzero coordinates of $u$.

A \emph{linear binary code} of parameters $[n,k,d]_2$ is a $k$–dimensional subspace 
\(
C \subseteq \FF_2^n
\)
whose \emph{minimum distance}
\[
d := \min_{\substack{u,v \in C \\ u \ne v}} \rho(u,v)
\]
is the smallest Hamming distance between two distinct codewords in $C$.
Equivalently, since $C$ is linear, the minimum distance can be expressed as
\[
d = \min_{x \in C \setminus \{0\}} w(x).
\]

For a linear code $C \subseteq \FF_2^n$ and any vector $u \in \FF_2^n$, 
the \emph{weight of the coset} $u + C$ is defined by
\[
w(u + C) := \min_{x \in u + C} w(x).
\]

For every linear $[n,k,d]_2$–code $C$, there exists a matrix 
$
H \in M_{(n-k)\times n}(\FF_2),
$
of rank $n-k$ called a \emph{check matrix} such that
$
C = \{\,x \in \FF_2^n : Hx^T = 0\,\}.
$
A linear code can correct up to $e$ errors if and only if any $2e$ columns of $H$ are linearly independent~\cite[Ch.~1]{MacWilliamsSloane}, \cite{ref1}.

\begin{lemma}
\(
N(k,e) \;\le\; N(k+1,e).
\)
\end{lemma}

\begin{proof}
Consider an optimal code 
\(
F_{k+1}: \FF_2^{k+1} \;\longrightarrow\; \FF_2^{N(k+1,e)}
\)
that corrects $e$ errors. Then the code 
\[
F_k: \FF_2^k \;\longrightarrow\; \FF_2^{N(k+1,e)}, 
\qquad
(x_1,\dots,x_k)\;\mapsto\; F_{k+1}(x_1,\dots,x_k,0),
\]
also corrects $e$ errors. Hence $N(k,e)\le N(k+1,e)$.
\end{proof}

There exists a well-known estimate called the Hamming bound. All balls of radius $e$ centered at distinct codewords are required to be disjoint~\cite{Hamming1950,ref1}:
\[
\sum_{i=0}^{e} \binom{n}{i} \, 2^k \;\le\; 2^n.
\]
In particular, for the single–error–correcting case ($e=1$), this reduces to
\[
(1 + n)\, 2^k \;\le\; 2^n
\quad\Longleftrightarrow\quad
k \;\le\; n - \log_2(n+1).
\]
The smallest integer $n$ satisfying this inequality for given $k$ and $e$ is called the \emph{Hamming bound} and is denoted by $\mathrm{Ham}(k,e)$, formally defined as
\[
\mathrm{Ham}(k,e)
\;:=\;
\min\Bigl\{\,n \in \NN \;\Bigm|\;
\sum_{i=0}^{e} \binom{n}{i}\, 2^k \;\le\; 2^n
\Bigr\}.
\]

\section{The Single‐Error‐Correcting Case}

\begin{proposition}
The Hamming bound for one error is attained, that is, \(
\mathrm{Ham}(k,1) \;=\; N(k,1).
\)
\end{proposition}

\begin{proof}
Let $n$ be the codeword length such that $2^{\,n-k} > n$. We construct a linear code with check matrix 
\[
H \;\in\; M_{(n-k)\times n}(\FF_2),
\]
of rank $n-k$, in which any two columns are linearly independent. Over $\FF_2$, this means there is no zero columns and no two identical columns.

Place the standard basis vectors in the first $n-k$ columns. Each remaining column is filled with a nonzero vector not yet used. This is possible because $2^{\,n-k} > n$ implies there are more nonzero vectors than columns. The rank of $H$ is $n-k$ thanks to the first $n-k$ columns.

Hence the only condition on $n$ is
\[
2^{\,n-k} > n 
\;\iff\; 
2^{\,n-k} \ge n+1 
\;\iff\; 
n-k \;\ge\; \log_2(n+1) 
\;\iff\; 
k \;\le\; n \;-\;\log_2(n+1).
\]
The Hamming bound is necessary for the existence of such a code, and this construction shows it is sufficient.
\end{proof}

Therefore,
\[
N(k,1) \;=\; \min\bigl\{\,n\in\NN \;\bigm|\; n - \log_2(n+1)\;\ge\; k \bigr\}.
\]

\begin{lemma}{\label{lemma:delta}}
The following statements hold
\begin{enumerate}
\item[(a)] $N(1, 1) = 3$;
\item[(b)] $1 \leq N(k + 1, 1) - N(k, 1) \leq 2$;
\item[(c)] $N(k + 1, 1) - N(k, 1) = 2 \iff k = 2^s - s - 1$ for some $s \geq 2$.
\end{enumerate}
\end{lemma}

\begin{proof}
\textbf{(a)} For $k=1$, the minimal $n$ satisfying the Hamming bound $\log_2(n+1) + 1 \leq n$ is as follows. For $n=2$: $1 + \log_2 3 > 2$, for $n=3$: $2 + \log_2 4 = 4 \leq 4$ which holds.

\medskip

Next denote $n = N(k, 1)$. \\
\textbf{(b)} We perform a technical check:
\begin{alignat*}{2}
&(n - 1) - \log_2 n < k 
  &&\implies n - \log_2 n < k + 1 \\
&\quad &&\implies n - \log_2(n + 1) < k + 1 \\
&\quad &&\implies N(k+1, 1) > n \\
&\quad &&\implies N(k + 1, 1) \geq N(k, 1) + 1.
\end{alignat*}
On the other hand:
\begin{align*}
n - \log_2(n + 1) \geq k 
  &\implies n + 2 - \log_2(n + 1) \geq k + 2 \\
  &\implies n + 2 - \log_2(n + 3) \geq k + 1 \\
  &\implies N(k + 1, 1) \leq N(k, 1) + 2.
\end{align*}

\medskip

\noindent\textbf{(c)}  ($\Rightarrow$) Suppose $N(k + 1, 1) = n + 2$. Then
\[
n - \log_2(n + 1) \geq k, \quad (n+1) - \log_2(n+2) < k + 1.
\]
From the second inequality: $n < \log_2(n+2) + k$, whence
\[
\log_2(n+1) \leq n - k < \log_2(n+2) \implies n + 1 \leq 2^{\,n - k} < n + 2 \implies 2^{\,n - k} = n + 1
\]
Since we work only with integers and $n \geq N(1, 1) = 3$, it follows that $n = 2^s - 1$ for $s \geq 2$:
\[
2^{\,2^s - 1 - k} = 2^s \implies  2^s - 1 - k = s \implies k = 2^s - s - 1.
\]
($\Leftarrow$) Let $k = 2^s - s - 1$. Then $N(k, 1) = 2^s - 1$, since
\begin{align*}
(2^s - 1) - \log_2(2^s) &= 2^s - 1 - s = k, \\[5pt]
2^s - 2 - \log_2(2^s - 1) &< (2^s - 1) - \log_2(2^s) = k.
\end{align*}
The last remaining step is
\begin{align*}
2^s - \log_2(2^s + 1) &< 2^s - s = k + 1 \\
&\implies N(k + 1, 1) > N(k, 1) + 1 \\
&\implies N(k + 1, 1) - N(k, 1) = 2.
\end{align*}
\end{proof}

\begin{figure}[H]
    \centering
    \includegraphics[width=0.75\linewidth]{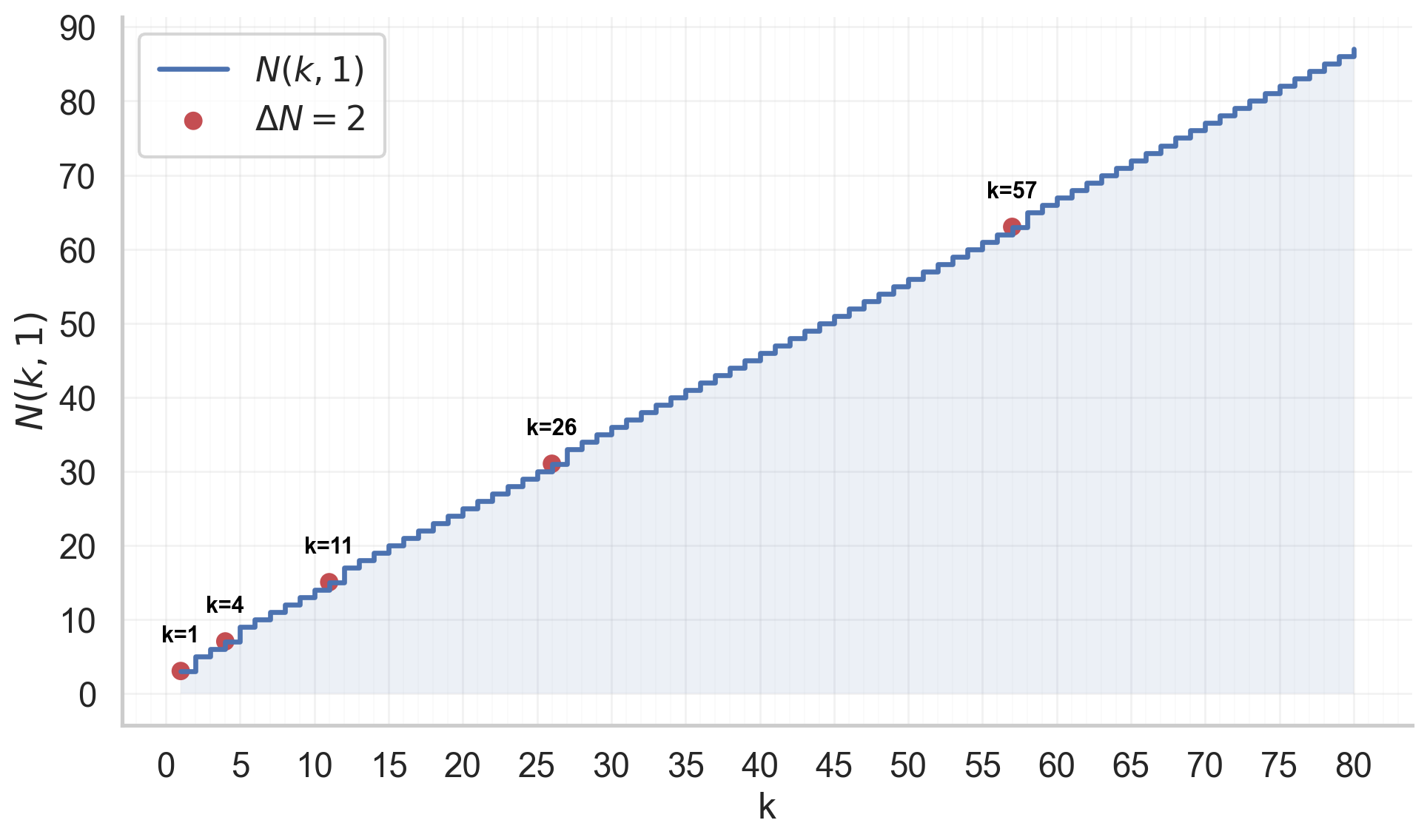}
    \caption*{Staircase plot of $N(k,1)$ for $1 \le k \le 80$. 
    Steps of height $2$ occur exactly at $k=2^s-s-1$ (Lemma~\ref{lemma:delta}).}
    \label{fig:staircase}
\end{figure}

\begin{theorem}\label{thm:main_formula}
The following formula holds:
\[
\mathrm{NL}(k, 1) = N(k, 1) = k + 1 + \left\lfloor \log_2\bigl(k + 1 + \lfloor \log_2 k \rfloor\bigr) \right\rfloor.
\]
\end{theorem}

\begin{proof}
As follows from part~(b) of Lemma~\ref{lemma:delta}, when $k$ increases by one, the value $N(k, 1)$ increases by 1 or 2. According to part~(c) of the same lemma, an increase by 2 occurs precisely for those $k$ of the form $k = 2^s - s - 1$ for some $s$. Thus,
\[
N(k, 1) = k + 1 + s(k),
\]
where $s(k)$ is the largest integer satisfying $2^{\,s(k)} - s(k) - 1 < k$.

We will show that
\[
s(k) = \left\lfloor \log_2\bigl(k + 1 + \lfloor \log_2 k \rfloor \bigr) \right\rfloor.
\]
For small values $k \leq 4$, the equality can be easily verified by direct computation. Henceforth assume $k \geq 5$, and fix such $s \geq 3$ that $s(k) = s$. Then $2^s - s - 1 < k \leq 2^{s+1} + s$. We consider two cases.

In the first case, when $2^s - s - 1 < k < 2^s$, we get $\lfloor \log_2 k \rfloor = s - 1$. Then
\[
k + 1 + \lfloor \log_2 k \rfloor \;\in\; [2^s,\, 2^s + s - 1]
\;\implies\;
\log_2\bigl(k + 1 + \lfloor \log_2 k \rfloor\bigr) \;\in\; [s,\, \log_2(2^s + s - 1)].
\]
Since $s \geq 3$, $\log_2(2^s + s) < s + 1$, hence
\[
\left\lfloor \log_2\bigl(k + 1 + \lfloor \log_2 k \rfloor \bigr) \right\rfloor = s.
\]

In the second case, if $2^s \leq k \leq 2^{s+1} + s$, then $\lfloor \log_2 k \rfloor = s$, and
\[
k + 1 + \lfloor \log_2 k \rfloor \;\in\; [2^s + s + 1,\, 2^{s+1} + 2s + 1].
\]
Noting that $2^{s+1} > 2^s + 2s + 1$ for $s \geq 3$, we conclude
\[
\left\lfloor \log_2\bigl(k + 1 + \lfloor \log_2 k \rfloor \bigr) \right\rfloor = s.
\]

Hence, in both cases, the value $s(k)$ coincides with the right‐hand side of the formula in the theorem, which completes the proof.
Since the constructed code is linear, it follows that $\mathrm{NL}(k,1) = N(k,1)$.
\end{proof}

\section{The General Case}

\begin{lemma}\label{lemma:basic_values}
Let $C$ be a linear $[n,k,2e+1]_2$‐code. Then
\begin{enumerate}[label=(\alph*)]
  \item vectors of weight $0$, $1$, $\ldots$, $e$ belong to distinct cosets modulo $C$;
  \item cosets of weight $0$, $1$, $\ldots$, $e - 1$ do not contain any vector of weight $e + 1$;
  \item any coset contains at most $\bigl\lfloor \tfrac{n}{e + 1}\bigr\rfloor$ distinct vectors of weight $e+1$;
  \item a coset of weight $e$ contains at most $\bigl\lfloor \tfrac{n-e}{e+1}\bigr\rfloor$ distinct vectors of weight $e+1$;
\end{enumerate}
\end{lemma}

\begin{proof}
\textbf{(a)} By contradiction, that two distinct vectors $u_1,u_2$ of weights $\le e$ lie in the same coset: $u_1\equiv u_2\pmod C$, so $c:=u_1-u_2\in C\setminus\{0\}$. Then 
\[
w(c)=w(u_1+u_2)\le w(u_1)+w(u_2)\le 2e,
\]
contradicting that the minimum distance is $2e + 1$.

\medskip
\noindent\textbf{(b)} Suppose $u$ has weight $\le e-1$ and $v$ has weight $e+1$, and $v\equiv u\pmod C$. Then $c=v-u\in C\setminus\{0\}$ has 
\[
w(c)\le w(v)+w(u)<e+1+e=2e+1,
\]
again contradicting $d=2e+1$. Hence no coset of weight $0$, $1$, $\ldots$, $e-1$ can contain a vector of weight $e+1$.

\medskip
\noindent\textbf{(c)} Let $v_1,v_2$ be two distinct vectors of weight $e+1$ in the same coset: $v_1\equiv v_2\pmod C$. Then $c=v_1-v_2\in C\setminus\{0\}$ has 
\[
w(c)\le w(v_1)+w(v_2)=2e+2,
\]
but since $v_1$ and $v_2$ are both of weight $e+1$, if they shared any support bit, then 
\[
w(c)\le e+1+e+1-2=2e<2e+1,
\]
contradicting $d=2e+1$. Therefore any two weight‐$(e+1)$ vectors in the same coset should have disjoint supports. Since each has support of size $e+1$, there can be at most $\lfloor n/(e+1)\rfloor$ of them in a single coset.

\medskip
\noindent\textbf{(d)} Fix a vector $u$ of weight $e$, and consider any vector $v$ of weight $e+1$ in the same coset: $v=u+c$ for some $c\in C$. Since $d=2e+1$, $c$ should have weight $\ge e+1$. If $v$ and $u$ shared any nonzero coordinate, then $c=v-u$ would have weight 
\[
w(c)\le w(v)+w(u)-2 = e+1+e-2=2e-1<2e+1,
\]
a contradiction. Thus $v$ should have its $e+1$ nonzero positions outside the two nonzero positions of $u$, leaving $n-e$ available coordinates. By the same disjoint‐support argument as in (c), there can be at most $\lfloor (n-e)/(e+1)\rfloor$ weight‐$(e+1)$ vectors in that coset.
\end{proof}

\begin{theorem}\label{th:tilted}
In any linear $[n,k,2e+1]_2$‐code, the following inequality holds:
\[
\left\lfloor \frac{n}{\,e+1\,} \right\rfloor
\!\left( 2^{\,n-k}
\;-\;
\sum_{i=0}^{e} \binom{n}{i}
\right)
\;+\;
\left\lfloor \frac{n-e}{\,e+1\,} \right\rfloor
\binom{n}{e}
\;\ge\;
\binom{n}{e+1}.
\]
\end{theorem}

\begin{proof}
There are $2^{\,n-k}$ cosets of $C$ in $\mathbb{F}_2^n$. 
By Lemma~\ref{lemma:basic_values}(a), the vectors of weights $0,1,\ldots,e$ lie in distinct cosets, and there are
\[
\sum_{i=0}^{e} \binom{n}{i}
\]
such vectors. Hence the number of remaining cosets with no vector of weight $\le e$ is
\[
2^{\,n-k} - \sum_{i=0}^{e} \binom{n}{i}.
\]
By Lemma~\ref{lemma:basic_values}(c), each such coset contains at most 
$\lfloor n/(e+1)\rfloor$ vectors of weight $e+1$.
Meanwhile, each coset of weight $e$ can contain at most 
$\lfloor (n-e)/(e+1)\rfloor$ vectors of weight $e+1$ (Lemma~\ref{lemma:basic_values}(d)), 
and there are $\binom{n}{e}$ cosets of weight $e$.
Therefore, the total number of weight‐$(e+1)$ vectors covered by all cosets is at most
\[
\left\lfloor \frac{n}{\,e+1\,} \right\rfloor
\!\left( 2^{\,n-k} - \sum_{i=0}^{e} \binom{n}{i} \right)
\;+\;
\left\lfloor \frac{n-e}{\,e+1\,} \right\rfloor \binom{n}{e}.
\]
But the total number of all weight‐$(e+1)$ vectors in $\mathbb{F}_2^n$ is $\binom{n}{e+1}$. 
Hence, the stated inequality follows.
\end{proof}

\medskip
Define
\[
C(k,e)\;:=\;\min\Bigl\{\,n\in\NN \;\Bigm|\;
\left\lfloor \frac{n}{e+1} \right\rfloor
\Bigl(2^{\,n-k}-\!\sum_{i=0}^{e}\binom{n}{i}\Bigr)
+\left\lfloor \frac{n-e}{e+1}\right\rfloor\binom{n}{e}
\;\ge\;\binom{n}{e+1}\Bigr\},
\]
that is, the smallest $n$ (for fixed $k$ and $e$) for which the inequality of Theorem~\ref{th:tilted} holds.  
This gives a lower bound on the length of any linear code:
\[
\mathrm{NL}(k,e) \;\ge\; C(k,e).
\]

\medskip

\begin{table}[h]
\centering

\begin{minipage}{0.47\textwidth}
\centering
\caption{Instances where $C(k,2)>\mathrm{Ham}(k,2)$.}
\label{tab:C2}
\begin{tabular}{rcc}
\toprule
$k$ & $\mathrm{Ham}(k,2)$ & $C(k,2)$ \\
\midrule
14   & 22  & 23  \\
22   & 31  & 32  \\
78   & 90  & 91  \\
114  & 127 & 128 \\
345  & 361 & 362 \\
494  & 511 & 512 \\
1427 & 1447& 1448\\
\bottomrule
\end{tabular}
\end{minipage}
\hfill
\begin{minipage}{0.47\textwidth}
\centering
\caption{Instances where $C(k,3)>\mathrm{Ham}(k,3)$.}
\label{tab:C3}
\begin{tabular}{rcc}
\toprule
$k$ & $\mathrm{Ham}(k,3)$ & $C(k,3)$ \\
\midrule
17    & 29  & 30  \\
32    & 46  & 47  \\
43    & 58  & 59  \\
57    & 73  & 74  \\
75    & 92  & 93  \\
98    & 116 & 117 \\
127   & 146  & 147  \\
\bottomrule
\end{tabular}
\end{minipage}

\end{table}

For most values of $k$, the number $C(k,e)$ coincides with $\mathrm{Ham}(k,e)$.
However, for certain parameters $(k,e)$ the coset–based bound of Theorem~\ref{th:tilted} is slightly stronger.
Tables~\ref{tab:C2} and~\ref{tab:C3} show several examples where $C(k,e)>\mathrm{Ham}(k,e)$.

\end{document}